\pdfoutput=1

\documentclass[11pt]{article}
\usepackage[margin=1in]{geometry}
\usepackage{amsthm}
\usepackage{amsfonts}
\usepackage{amssymb}
\usepackage{amsmath}
\usepackage{mathrsfs}
\usepackage{float}
\usepackage{caption}
\usepackage{subcaption}
\usepackage[newenum]{paralist}
\usepackage[usenames,dvipsnames]{xcolor}
\usepackage[noend]{algorithmic}
\usepackage{algorithm}
\usepackage[T1]{fontenc}
\usepackage{verbatim}
\usepackage[bottom]{footmisc}
\usepackage{calligra}
\usepackage{setspace} 
\usepackage{boxedminipage}
\usepackage{multirow}
\usepackage{graphicx}
\usepackage{hyperref}
\algsetup{linenodelimiter=.}
\algsetup{indent=2em}
\newlength\myindent
\newcommand\bindent{%
    \begingroup
    \setlength{\itemindent}{\myindent}
    \addtolength{\algorithmicindent}{\myindent}
}
\newcommand\eindent{\endgroup}

 \makeatletter
 \newtheorem*{rep@theorem}{\rep@title}
 \newcommand{\newreptheorem}[2]{%
 \newenvironment{rep#1}[1]{%
 \def\rep@title{#2 \ref{##1}}%
 \begin{rep@theorem}}%
 {\end{rep@theorem}}}
 \makeatother

 \newtheorem{theorem}{Theorem}
 \newtheorem{lemma}[theorem]{Lemma}
 
 \newtheorem{proposition}[theorem]{Proposition}
 \newtheorem{definition}[theorem]{Definition}
\newtheorem{operation}[theorem]{Operation}

 \newreptheorem{theorem}{Theorem}
 \newreptheorem{lemma}{Lemma}
 \newreptheorem{corollary}{Corollary}
\hypersetup{colorlinks,linkcolor=blue,filecolor=blue,citecolor=blue, urlcolor=blue,pdfstartview=FitH}

\DeclareMathOperator{\degree}{degree}
\DeclareMathOperator*{\E}{\mathbf{E}}
\newcommand{\lra}{Lenzen's routing algorithm}

\begin{document}
\title{Minimum-weight Spanning Tree Construction in $O(\log \log \log n)$ Rounds on the Congested Clique
  \thanks{This work is supported in part by National Science Foundation grant CCF-1318166.
  }}
\author{Sriram V. Pemmaraju \hspace{3em} Vivek B. Sardeshmukh\\ \small{Department of Computer Science, The University of Iowa, Iowa City, IA 52242}\\ 
\texttt{\{sriram-pemmaraju, vivek-sardeshmukh\}@uiowa.edu}}
\maketitle
\begin{abstract}
  This paper considers the \textit{minimum spanning tree (MST)} problem in the Congested Clique model and presents an algorithm that runs in $O(\log \log \log n)$ rounds, with high probability. 
  Prior to this, the fastest MST algorithm in this model was a deterministic algorithm due to Lotker et al.~(SIAM J on Comp, 2005) from about a decade ago.
  A key step along the way to designing this MST algorithm is a \textit{connectivity verification} algorithm that not only runs in $O(\log \log \log n)$ rounds with high probability, but also has low message complexity. 
  This allows the fast computation of an MST by running multiple instances of the connectivity verification algorithm in parallel.

  These results depend on a new edge-sampling theorem, developed in the paper, 
  that says that if each edge $e = \{u, v\}$ is sampled independently with probability $c \log^2 n/\min\{\mbox{degree}(u), \break \mbox{degree}(v)\}$ (for a large enough constant $c$) then
  all cuts of size at least $n$ are approximated in the sampled graph. 
  This sampling theorem is inspired by series of papers on graph sparsification via random edge sampling due to Karger~(STOC 1994),
  Bencz\'{u}r and Karger~(STOC 1996, arxiv 2002), and Fung et al.~(STOC 2011).
  The edge sampling techniques in these papers use probabilities that are functions of edge-connectivity or a related measure called edge-strength.
  For the purposes of this paper, these edge-connectivity measures seem too costly to compute and the main technical contribution of this paper is to
  show that degree-based edge-sampling suffices to approximate large cuts.
\end{abstract}

\section{Introduction}
The $\mathcal{CONGEST}$ model is a synchronous, message-passing model of distributed
computation in which the amount of information that a node can transmit along
an incident communication link in one round is restricted to $O(\log n)$ bits, where
$n$ is the size of the network \cite{peleg2000distributed}. As the name suggests, the
$\mathcal{CONGEST}$ model focuses on congestion as an obstacle to distributed computation. 
In this paper, we focus on the design of distributed algorithms in the
$\mathcal{CONGEST}$ model on a clique communication network; we call this the
\textit{Congested Clique} model. In the Congested Clique model, all information is nearby, i.e., 
at most one hop away, and so any difficulty in solving a problem is due to congestion alone
In this paper we focus on the \textit{minimum spanning tree (MST)} problem in the Congested
Clique model and show how to solve it in $O(\log \log \log n)$ rounds with high probability.
Prior to this, the fastest MST algorithm in the Congested Clique was a deterministic algorithm 
due to Lotker et al.~\cite{lotker2005mstJournal} from about a decade ago.

The MST problem has a long history in distributed computing~\cite{gallager1983ghs, awerbuch1987optimal,garay1998sublinear, KhanTheoreticalCS2007}.
After a long sequence of results on MST through the 80's and 90's, Kutten and Peleg 
\cite{KuttenPeleg1998} showed how to compute an MST in the $\mathcal{CONGEST}$ model in $O(D + \sqrt{n} \cdot \log^* n)$ 
rounds on $n$-vertex diameter-$D$ graphs.
The near-optimality of this result was established by lower bounds on MST construction in the 
$\mathcal{CONGEST}$ model due to Peleg and Rubinovich \cite{PelegRubinovich2000}, Elkin \cite{Elkin2006}, and most recently
Das Sarma et al.~\cite{DasSarmaSICOMP2011}.
In the latter paper \cite{DasSarmaSICOMP2011}, a lower bound of $\Omega(\sqrt{n/\log^2 n} + D)$ is shown for 
$D = \Omega(\log n)$.
Lower bounds are known for smaller $D$ as well; for example, for $D = 3$, Das Sarma et al.~\cite{DasSarmaSICOMP2011}
show a lower bound of $\Omega((n/\log n)^{1/4})$.
Note that there are no lower bounds known for $D = 2$ or $D = 1$, which is the setting we are 
interested in.
For diameter-1 graphs, i.e., cliques, the $O(\log \log n)$-round deterministic algorithm of 
Lotker et al.~\cite{lotker2005mstJournal} has been the fastest known for more than a decade.
The lack of lower bounds in the Congested Clique model has kept open the possibility that faster
MST algorithms are possible and we show that this in indeed the case by presenting an exponentially
faster algorithm.

A key step along the way to designing the above-mentioned MST algorithm is a 
\textit{connectivity verification} algorithm in the Congested Clique model that not only 
runs in $O(\log \log \log n)$ rounds with high probability, but also has low message complexity. 
The low message complexity allows the fast computation of an MST by running multiple instances 
of the connectivity verification algorithm in parallel.
These results depend on a new edge-sampling theorem, developed in the paper, 
that says that if each edge $e = \{u, v\}$ is sampled independently with probability $c \log^2 n/\min\{\mbox{degree}(u), \mbox{degree}(v)\}$ (for a large enough constant $c$) then
all cuts of size at least $n$ are approximated in the sampled graph.
This sampling theorem is inspired by series of papers on graph sparsification via random edge 
sampling due to Karger~\cite{karger1994stoc}, %TODO: Vivek: Math.~of OR 1999??
Bencz\'{u}r and Karger~\cite{benczurKarger2002arxiv, benczurKarger1996stoc}, and Fung et al.~\cite{fung2011stoc}.
The edge sampling techniques in these papers use probabilities that are functions of 
edge-connectivity or a related measure called edge-strength.
For the purposes of this paper, these edge-connectivity measures seem too costly to compute 
and the main technical contribution of this paper is to
show that degree-based edge-sampling suffices to approximate large cuts.

\subsection{Main Results}
\label{section:mainResults}
In this paper, we achieve the following results:
\begin{itemize}
  \item We show how to solve the Connectivity Verification problem in $O(\log \log \log n)$ rounds  on a Congested Clique w.h.p.\footnote{We say an event occurs with high probability (w.h.p.), if the probability of that event is at least $(1-1/n^c)$ for a constant $c\geq 1$.} on an input graph $G$.
It has the following implication.

\item We show how to use this Connectivity Verification algorithm solve the MST problem in $O(\log \log \log n)$ rounds w.h.p. on a Congested Clique. 
\end{itemize}
In order to achieve our results, we use a variety of techniques that balance bandwidth constraints with the need to make rapid progress. 
One of the key technique we use is random edge sampling. 
In the next subsection we describes these sampling techniques at a high level. 
We believe that our techniques will have independent utility in any distributed setting in which congestion is a bottleneck. 

\subsection{Random Sampling in the Congested Clique Model}
\paragraph{Random graph sampling.} 
Randomly sampling vertices or edges to obtain a reduced-sized subgraph of the input graph has been studied in various computational models for a variety of problems. 
For example, in the sequential setting (RAM model) cut, flow, and network design problems can be solved faster on the sampled subgraph than on the input graph and 
more importantly, due to properties of the random sample, the solution on the sampled subgraph can be efficiently translated into a solution of the original graph~\cite{karger1994stoc, benczurKarger1996stoc, fung2011stoc}. 
Having a reduced-sized subgraph also enables solving problems efficiently in the streaming model~\cite{ahn2009streaming}.
Recently, applications of random vertex- and edge-sampling to solve problems in MapReduce model~\cite{KarloffSuriVassilvitskii} have been shown~\cite{lattanzi2011filtering}. 
% Hence, our sampling techniques can be of independent interest rather than just limited to Congested Clique model. 

\paragraph{Random sampling in the Congested Clique model.}
The Congested Clique model has high bandwidth availability over the entire network, but congestion at individual nodes.
Each node can communicate $\Theta(n)$ messages in each round and hence a total of $\Theta(n^2)$ messages are exchanged in a round over the entire network. 
Hence, an $n$-vertex graph can be fully communicated across the network in one round, but only a linear-sized subgraph can reach a single node. 
Given this situation, a general approach would be to use the outcome of local processing of a linear-sized subgraph to compute the solution to the original problem.  
Of course, the key challenge of such an approach is showing that a linear-sized subgraph with the appropriate properties can be quickly sampled. 
The question is how to set the probabilities of sampling such that (i) we get a linear-size subgraph and (ii) processing this subgraph enables efficient computation of the solution to the original problem. 
One can produce this linear-sized subgraph in a variety of ways - 
by randomly sampling vertices independently with probability $1/\sqrt{\alpha}$, where $\alpha$ is the average degree of $G$ or 
by sampling each edge independently with probability  $1/\alpha$, etc. 
But for many cases, the sampling-probabilities based on these ``local'' quantities such as degree, max-degree, average degree, etc are not adequate to obtain a subgraph with the appropriate properties. 
On the other hand, computing sampling probabilities which are based on ``global'' properties of the graph may be as hard as solving the problem on the original graph itself. 
In summary, the challenges of random sampling in the Congested Clique are two-fold - 
(i) how to set the sampling probabilities so that we a get a linear-sized subgraph with the appropriate properties and (ii) how to compute these probabilities quickly on the Congested Clique model.

In this paper, we describe our edge-sampling technique to solve the Connectivity Verification problem. 
Random edge-sampling has been shown to be useful in the context of cut-approximation in the sequential setting (RAM model)~\cite{karger1994stoc, benczurKarger1996stoc, fung2011stoc}. 
For example, Fung et al.~\cite{fung2011stoc} showed that edge-connectivity-based probabilities produce a $O(n\log^2 n)$-size subgraph which approximates all the cuts in the original graph w.h.p..  
Hence, the Connectivity Verification problem can be solved on this reduced-size sampled graph obtained in this way.
The problem with this approach is the probability function depends on the edge-connectivity which is a global-property and hence it might be difficult to compute quickly in the Congested Clique model. 
On the other hand, for the Connectivity Verification problem we don't need such a strong result on the cut-approximation similar to the result of Fung et al.~\cite{fung2011stoc}. 
In this report, we show that degree-based edge-sampling probabilities are sufficient to solve the Connectivity Verification problem. 
Specifically, we show the following result: 
\begin{itemize}
  \item If each edge $e \in E$ is independently sampled with probability based on the degrees of its end-points then the set of sampled edges $\hat{E}$ has the following properties w.h.p.:
    (i) $|\hat{E}| = O(n \log^2 n)$ and (ii) the number of inter-component\footnote{Refer to Subsection~\ref{sub:tech} for the definition of inter-component edges}  edges between components induced by $\hat{E}$ is $O(n)$. 
    (For the precise statement of this, see Theorem~\ref{thm:sample}).
\end{itemize}

\subsection{Preliminaries}
\label{sub:tech}
\paragraph{Maximal spanning forest and component graph.}
A \textit{maximal spanning forest} of a graph $G$ is a spanning forest of $G$ which has exactly as many trees as the number of components in $G$.    
  For a given graph $G(V,E)$ and a given subset of edges $\hat{E} \subseteq E$, let $\mathcal{C}$ be the set of connected components of graph $(V,\hat{E})$. 
  The \textit{component graph} ``induced'' by edges in $\hat{E}$ is the graph $cg[G, \hat{E}]$ whose vertices are components $\mathcal{C}$ and whose edges are 
 \textit{inter-component edges} defined as  
 \[E' = \left\{\{C_i, C_j\} \mid C_i, C_j \in \mathcal{C} \mbox{ and there exists } \{u,v\} \in E \mbox{ such that } u \in C_i, v \in C_j\right\}\]
where $C_i$ is the minimum of ID of nodes in component $C_i$. 
The minimum ID node in a component $C$ is also referred as the \textit{leader} of $C$ and denoted as $\ell(C)$.
For the convenience we \textit{label} a component $C$ by the ID of $\ell(C)$ and 
let $c(u)$ denote the label of the component of a node $u$. 
Refer to Figure~\ref{fig:componentGraph} for an illustration of a component graph. 
We can extend the concept of the component graph to weighted graphs as well by defining weights of inter-component edges as follows:
\[wt(C_i, C_j) = \min\{wt(u, v) \mid \{u, v\} \in E \mbox{ such that } u \in C_i, v \in C_j\}. \]
\begin{figure}
  \begin{boxedminipage}{\textwidth}
    \centering
    \begin{subfigure}[b]{0.3\textwidth}
     \includegraphics[scale=0.6]{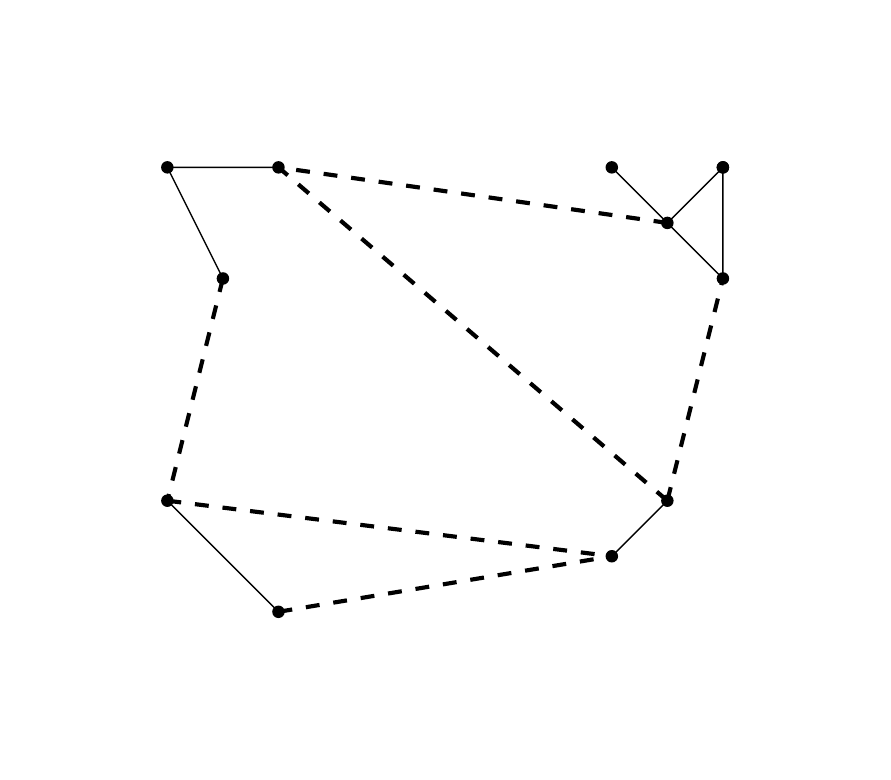} 
      \caption{Edges in $\hat{E}$\label{subfig:edges}}
    \end{subfigure}
    \begin{subfigure}[b]{0.3\textwidth}
     \includegraphics[scale=0.6]{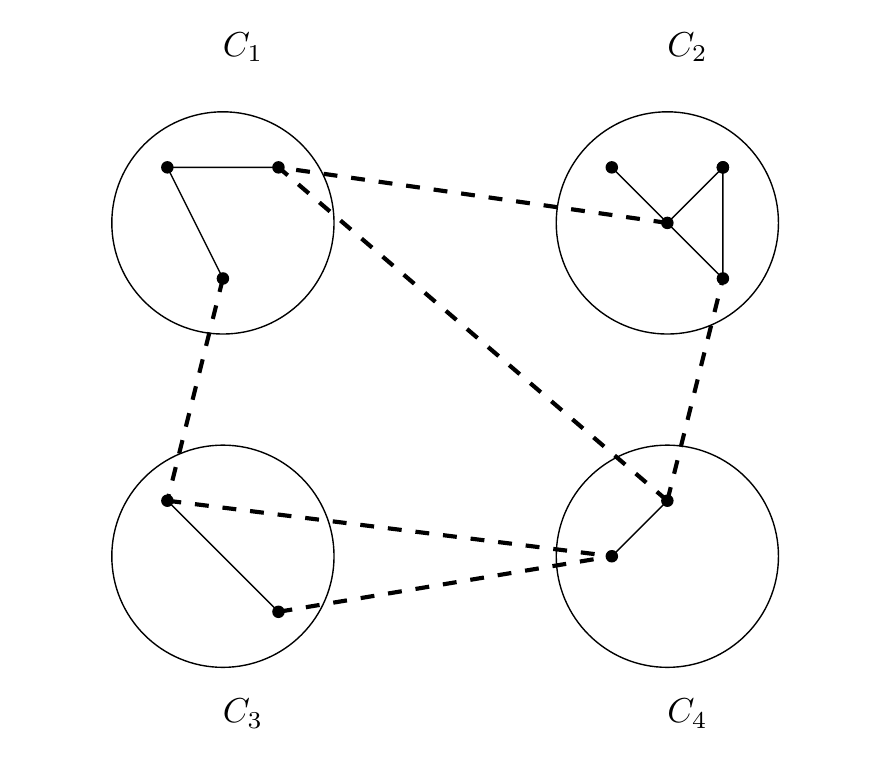} 
     \caption{Components induced by $\hat{E}$\label{subfig:component}}
    \end{subfigure}
    \begin{subfigure}[b]{0.3\textwidth}
     \includegraphics[scale=0.6]{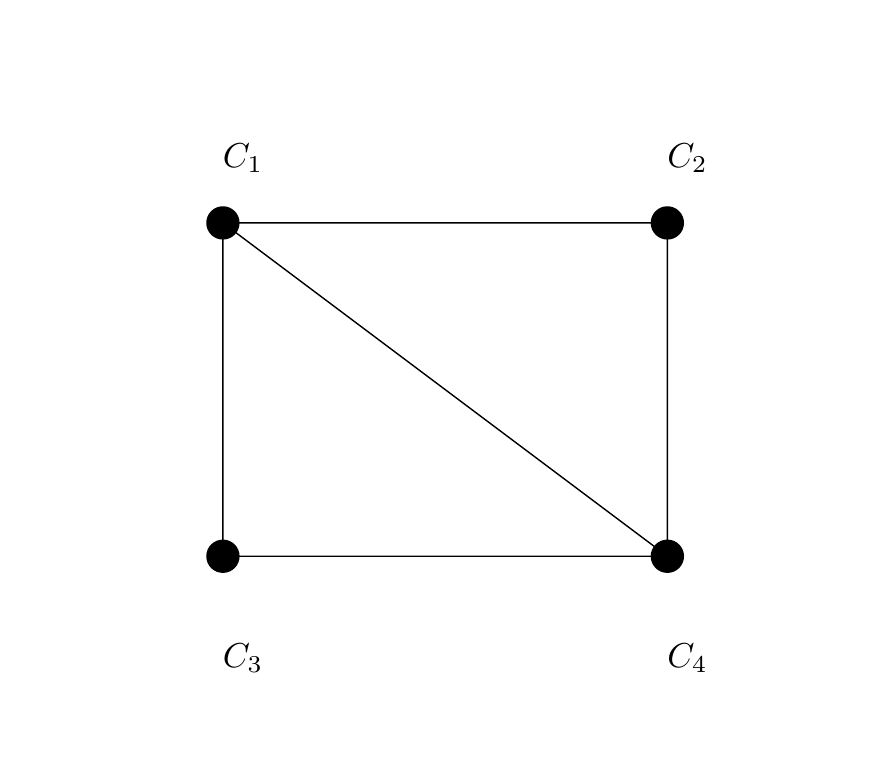}
      \caption{Component graph\label{subfig:cg}} 
    \end{subfigure}
  \end{boxedminipage}
  \caption{(a) The solid lines are the edges in $\hat{E}$ and the dashed lines are the edges in $E\setminus \hat{E}$ (b) The circles are the components induced by edges in $\hat{E}$ (c) The component graph induced by edges in $\hat{E}$ and the solid lines are the inter-component edges. 
  \label{fig:componentGraph}}
\end{figure}
\paragraph{Connectivity Verification problem.}
The input to the Connectivity Verification problem is a graph $G=(V, E)$ and the goal is to decide whether $G$ is connected or not. 
Initially, each node $v \in V$ knows incident edges in $E$. 
When the algorithm ends, all nodes in $V$ are required to know a maximal spanning forest of $G$ and hence can decide whether $G$ is connected or not.

\paragraph{MST problem.}
The input to MST problem is a weighted clique graph $G(V,E)$. 
(This can be generalized to any non-clique graph where weights of non-edges is set to $\infty$.)
Initially, each node $v \in V$ knows weights $w(v, w)$ to all nodes $w \in V$. 
When the algorithm ends, all nodes in $V$ are required to know a spanning tree $T$ of $V$ of minimum weight.
We assume that weights can be represented in $O(\log n)$ bits. 

\paragraph{Lenzen's routing protocol.}
A key algorithmic tool that allows us to design near-constant-time round 
algorithms is a recent deterministic routing protocol by Lenzen 
\cite{lenzen2013routing} that disseminates a large volume of information 
on a Congested Clique in constant rounds.
The specific routing problem, called an \textit{Information Distribution Task},
solved by Lenzen's protocol is the following.
Each node $i \in V$ is given a set of $n' \le n$ messages, each of size $O(\log n)$, $\{m_i^1, m_i^2, \ldots, m_i^{n'}\}$,
with destinations $d(m_i^j) \in V$, $j \in [n']$.
Messages are globally lexicographically ordered by their source $i$, destination $d(m_i^j)$, and $j$.
Each node is also the destination of at most $n$ messages.
Lenzen's routing protocol solves the Information Distribution Task in $O(1)$ rounds.

\section{Large-Cut-Preserving Random Edge-Sampling}
\label{sec:cut}
In this section, we describe how to sparsify the given graph $G$ by sampling edges such that ``large'' cuts of $G$ have approximately the same value in $G$ as in the sampled graph. 
The problem of approximating every cut of $G$ arbitrarily well in the sampled graph first introduced by Karger~\cite{karger1994stoc}. 
He showed that if the graph has minimum cut-size $c$, then sampling edges with probability roughly $1/{\epsilon^2 c}$ yields a graph with cuts that are all, with high probability, within $(1\pm \epsilon)$ of their expected values. 
However if $G$ has $m$ edges, then the sampled graph has $O(m/c)$ edges so this scheme may not sparsify the graph effectively when $c$ is small. 
Bencz\'{u}r and Karger~\cite{benczurKarger2002arxiv, benczurKarger1996stoc} later showed that one can obtain a weighted graph with $O(n \log n)$ edges which approximates all the cuts in the original graph. 
This was achieved by sampling edges with non-uniform probabilities which are based on \textit{edge-strengths} (a measure of edge-connectivity). 
Fung et al.~\cite{fung2011stoc} simplified the computation of these probabilities based on standard edge-connectivity and showed that one can obtain a weighted graph graph with $O(n \log^2 n)$ edges which approximates all the cuts in the original graph.  
The motivation behind the work of Bencz\'{u}r-Karger~\cite{benczurKarger1996stoc, benczurKarger2002arxiv} and Fung et al.~\cite{fung2011stoc} was to speed-up the computation of max-flow approximation, since the sampled graph has fewer edges, minimum cuts can be found in it faster than in the original graph. 
We are interested in ``preserving'' only large cuts, i.e., cuts of size $\Omega(n)$ as opposed to arbitrary cuts and hence a special case of above problems. 
In this section, we show that simple degree-based sampling probabilities are good enough to obtain a sparse graph with $O(n \log^2 n)$ edges which has at least one edge from every $\delta$-size cuts in $G$ when $\delta \geq n$. 
We start with defining few terms which will enable us to state the sampling probabilities.  

\begin{definition}[rounded-degree, $k$-degree edge]
  The rounded-degree of a vertex $u$ is defined as $rd(u) = 2^{\lfloor \log \degree(u) \rfloor}$. 
  The rounded-degree of an edge $e = \{u, v\}$ is defined as $rd(e) = \min\{rd(u), rd(v)\}$. 
  An edge $e$ is called a $k$-degree edge if its rounded-degree $rd(e) = k$. 
\end{definition}

An edge $e$ is independently sampled with probability $p_e = \min \{50\log^2n/rd(e),  1\}$. 
When edges of $G$ are sampled in this manner, we obtain a graph $\hat{G}$, such that w.h.p. (i) $\hat{G}$ has $O(n \log^2 n)$ edges and (ii) every cut of size at least $n$ in $G$ contains at least one edge in $\hat{G}$.
We can obtain approximate cuts in the sampled graph for large cuts in $G$ if we assign weights $1/p_e$ to the sampled edges in the sampled graph, but for our purpose it is sufficient to show that at least one edge from every large cut is sampled. 
More formally, we'll prove the following theorem:
\begin{theorem}
  For a given undirected graph $G$, if we sample each edge $e$ independently with probabilities $p_e = \min\{ 50\log^2n/rd(e) , 1\}$ to obtain the sampled graph $\hat{G}$ then the following properties hold with high probability: 
  (i) Number of edges in $\hat{G} = O(n \log^2 n)$ and (ii) from every cut of size at least $n$ in $G$ at least one edge is sampled. 
\end{theorem}
\noindent We prove this theorem in a similar manner as Fung et al.~\cite{fung2011stoc} proved for edge-connectivity based probabilities. 
We first define $k$-projection of a cut and then count distinct $k$-projections of cuts of sizes $\delta$. 
Then we show that ``bad events'' happen with low probability for a single projection and then use the counting result to apply a union bound. 
\subsection{The Projection Counting Theorem}
\begin{definition}[$k$-projection]
  For a cut $C$ the subset $S \subseteq C$ of $k$-degree edges is called $k$-projection of $C$. 
\end{definition}
It is worth mentioning that this definition of $k$-projection is different from the definition of $k$-projection in Fung et al.~\cite{fung2011stoc}. 
They defined $k$-projection of a cut $C$ to be subset $S\subseteq C$ of edges whose edge-connectivity is at least $k$. 
We are sampling edges based on degrees as opposed to edge-connectivities and hence the definition of $k$-projection is also based on degrees. 

\noindent To count the number of distinct $k$-projections of cuts of size $\delta$, consider the following two operations.
\begin{operation}[splitting-off]
  \textit{`Splitting-off a pair of edges'} refers to replacing the pair of edges $\{s, u\}$ and $\{u, t\}$ in an undirected multigraph by a single edge $\{s, t\}$. 
  The operation \textit{`splitting-off a vertex'} with even degree in an undirected graph refers to splitting-off all arbitrary pairs of incident edges on this vertex.  
\end{operation}
This splitting-off operation was first introduced by Lov\'{a}z in~\cite{lovasz1993book}. 
Various properties of splitting-off operations have been studied in the past~\cite{mader1978reduction, babai1986complexity}. 
Mader~\cite{mader1978reduction} showed that there exists a pair of edges incident on an even-degree vertex such that splitting off this pair does not decrease connectivity of the rest of the edges in the resulting graph. 
We need a very simple property:
\begin{proposition} 
Splitting-off an even-degree vertex preserves vertex degrees of the rest of the vertices in the resulting graph. 
\end{proposition}

\noindent The following operation was introduced by Karger~\cite{karger1993soda}. 
\begin{operation}[edge-contraction]
  Contraction of an edge $e = \{u, v\}$ in an undirected multigraph is defined as merging $u$ and $v$ into a single vertex.
  Any self-loops produced by edges between $u$ and $v$ are discarded. 
  We call undoing contracted edge $e$ as \textit{expanding} the vertex formed by contracting $e$. 
\end{operation}

\noindent We now present the $k$-projection counting theorem.
The proof of this theorem is similar to the celebrated cut-counting result of Karger~\cite{karger1993soda} 
and the use of of the ``splitting-off'' operation is inspired by Fung et al.~\cite{fung2011stoc}.
We first describe a randomized algorithm which outputs a random cut.
This algorithm establishes a probability distribution over $k$-projections. 
We bound the probability of this random cut being the cut under inception and use this to bound number of distinct $k$-projections.
\begin{lemma}[$k$-projection counting] \label{lemma:projectionCount}
  Let $G$ be an undirected size-$n$ graph. 
  For any $k$ that is power of~$2$ and $\delta \geq 1$ the number of distinct $k$-projections of cuts of size at most $\delta$ is at most $n^{2\delta/k}$.
\end{lemma}
\begin{proof}
  We construct a multigraph $G_M$ by adding two copies of $G$, i.e., for each edge $e$ in $G$ we have two copies of $e$ in $G_M$. 
  There is a bijection between size-$\delta$ cuts in $G$ and size-$2\delta$ cuts in $G_M$. 
  Further, there is a bijection between $k$-projections of size-$\delta$ cuts in $G$ and $2k$-projections of size-$2\delta$ cuts in $G_M$.
  Hence, it suffices to prove above lemma for $2k$-projections of size-$2\delta$ cuts in $G_M$.
  To avoid carrying the ``2'' through the rest of the proof we prove the lemma for $k$-projections of size-$\delta$ cuts in $G_M$ where $k = 2^\ell, \ell\geq 1$ and $\delta$ is even.

  \noindent We run the following randomized algorithm on $G_M$.
We want this algorithm to output a $k$-projection and thereby establish a probability distribution on $k$-projections.
  \begin{enumerate}
    \item Split-off all vertices $u$ whose rounded degree $rd(u) < k$.
    \item Contract an edge chosen uniformly at random in the resulting graph.
    \item If the contraction operation produces a vertex $u$ with $rd(u) < k$, split it off. 
    \item If at most $2\delta/k$ vertices are left, generate a random cut and output $k$-degree edges from the cut (where the degrees are with respect to $G_M$); otherwise, go to Step 2.  
  \end{enumerate}
  Fix a cut $C$ of size at most $\delta$. 
  Let $S$ be its $k$-projection. 
  Since rounded degree of all the edges in $S$ is $k$, none of the edges in $S$ are split-off by Step~1. 
  We argue below that the probability that none of the edges in $S$ are contracted in the above algorithm is at least $1/n^{2\delta/k}$ and therefore, there can be at most $n^{2\delta/k}$ such different $k$-projections. 

  Observe that splitting-off vertices with rounded-degree $<k$ does not affect the degrees of the rest of the vertices. 
  Therefore, in the splitting-off process no edge from $S$ is split-off. 
  If no edge in $C$ is contracted then no edge from $S$ is contracted.
  We now bound the probability that an edge from $C$ is contracted. 
  Let $G_i$ be the multigraph and $h_i$ be the remaining vertices at the beginning of iteration $i$ of the contraction algorithm. 
%   Let $h_i$ be the remaining vertices at the beginning of iteration $i$ of the contraction algorithm.
  Note that $h_1$ is the number of vertices in $G_M$ with $rd(v) \geq k$. 
  At the start of iteration $i$, there are at least $kh_i/2$ edges in the current multigraph $G_i$.
  The size of cut $C$ does not increase during the splitting-off process. 
  Hence the probability that no edge in $C$ is selected to contract in Step 2 of iteration $i$ is at least $ 1 - \frac{\delta}{h_ik/2}$. 
  Therefore, the probability that no edge in $C$ is selected in the entire execution of the algorithm is at least
  \[
    \prod_i \left(1 - \frac{\delta}{h_ik/2}\right) \geq \prod_{j=n}^{2\delta/k+1}\left(1-\frac{\delta}{jk/2}\right) = {n \choose 2\delta/k}^{-1}.
  \]
  Once the number of vertices reaches $2\delta/k$, the algorithm generates a random cut. 
  Since there are at most $2^{2\delta/k - 1}$ distinct cuts in a graph with $2\delta/k$ vertices, the probability that the random cut generated by the algorithm contains only edges in $C$ is at least 
  ${n \choose 2\delta/k}^{-1}2^{1 - 2\delta/k} \geq n^{-2\delta/k}$. 
  Let the random cut generated by the algorithm is $(V_1, V_2)$. 
  We now backtrack the execution of the algorithm to ``expand'' vertices, i.e. we obtain the multigraph $G_{i}$ from the multigraph $G_{i+1}$ by expanding the vertex that was formed by the edge chosen for contraction in iteration $i$. 
  After expansion let $V_1$ becomes $U_1$ in $G_0$ (i.e. multigraph at the beginning of the contraction algorithm) and $V_2$ becomes $U_2$ in $G_0$. 
  We report the $k$-degree edges from the cut $(U_1, U_2)$ and the degrees are with respect to $G_0$. 
  (Note that the degree of a vertex in $G_0$ is the same as its degree in $G_M$ since splitting-off is degree-preserving operation). 
  Observe that if none of the edges in $C$ are contracted in the process then it means none of the edges in $S$ are contracted (some of the edges in $C$ might have split-off but edges in $S$ does not split-off). 
  Therefore, $S$ is exactly the set of $k$-degree edges in $G_M$ output by the algorithm. 
  This is true for every distinct $k$-projection of cuts having at most $\delta$ edges. 
  Hence the lemma follows. 
\end{proof}

\subsection{Sampling Large Cuts}
In this subsection, we show that our choice of degree-based sampling probabilities ``preserves''  $k$-projections given that these are large enough. 
The upper bound on the number of $k$-projections proved in Lemma~\ref{lemma:projectionCount} allows us to apply a union bound.
After showing that large $k$-projections survive in the sampled graph we show at least one edge from large cuts is sampled with high probability.
\begin{definition}[$\delta$-good $k$-projection]
  The $k$-projection $S$ of a cut $C$ of size $\delta$ is a $\delta$-good $k$-projection if $|S| \geq \delta/2\log n$.
\end{definition}

\begin{lemma} \label{lemma:goodProjection}
  Let $G=(V,E)$ be an $n$-vertex graph and let $\hat{E}$ be the set of edges obtained by independently sampling each edge $e\in E$ with probability $p_e = \min\{50\log^2 n/rd(e), 1\}$. 
  For any $\delta \geq n$ and any $k >1$ that is power of~$2$, with probability at least $1 - 1/n^3$ every $\delta$-good $k$-projection contains an edge in $\hat{E}$. 
\end{lemma}
\begin{proof}
  Fix a size $\delta \geq n$ and $k$. 
  Consider a $\delta$-good $k$-projection $S$, i.e., $S$ is a $k$-projection of size at least $\delta/2\log n$ of size-$\delta$ cut.
  Since all edges in $S$ are $k$-degree edges, the probability that an edge is sampled from $S$ is $\min\{50\log^2 n/k, 1\}$.
  If $k \leq 50\log^2 n$ then we are done, therefore assume that $k > 50\log^2 n$. 
  Let $X$ be the number of edges sampled from $S$.
  Then $\E[X] \geq 25\delta\log n/k$. 
  Since edges are sampled independently, by Chernoff's bound we have 
  \[ 
    \Pr\left(X \leq \frac{\E[X]}{5}\right)  \leq \frac{1}{n^{8\delta/k}}. 
  \]
  By Lemma~\ref{lemma:projectionCount}, for a fixed $\delta$ and $k$ there are at most $n^{2\delta/k}$ distinct $k$-projections. 
  Applying a union bound over all $\delta$-good $k$-projections the probability that 
  the there exists a $\delta$-good $k$-projection from which fewer than $5\delta\log n/k$ edges are sampled is at most 
     \[ \frac{n^{2\delta/k}}{n^{8\delta/k}} = \frac{1}{n^{6\delta/k}} \leq \frac{1}{n^6}. \]
  The last inequality is due to the fact that $k < n$ and the assumption that $\delta \geq n$. 
  Again applying a union bound over at most $n^2$ different values of $\delta$ and at most $\log n$ different values of $k$, for any $k$ and any $\delta \geq n$,
  the probability that there exists a $\delta$-good $k$-projection with fewer than $5\delta\log n/k$ edges sampled is at most 
  $ {n^3}/{n^6} = {1}/{n^3}.$
  Hence the lemma follows. 
\end{proof}

\begin{theorem}\label{thm:cut} 
  Let $G=(V,E)$ be an $n$-vertex graph and let $\hat{E}$ be the set of edges obtained by independently sampling each edge $e\in E$ with probability $p_e = \min\{50\log^2 n/rd(e), 1\}$. 
  Then with probability at least $1 - 1/n^3$ every cut of size at least $n$ contains an edge in $\hat{E}$. 
\end{theorem}
\begin{proof}
  Let $\mathscr{C}$ be the set of all cuts of size $\delta \geq n$. 
  If for some $k$ we can show there exists a $\delta$-good $k$-projection for every $C \in \mathscr{C}$, then by Lemma~\ref{lemma:goodProjection} we are done. 

  For every cut $C \in \mathscr{C}$, partition $C$ into $S_1, S_2, S_4, \ldots S_{2^{\lfloor \log n \rfloor}}$, where $S_k$ is the $k$-projection of $C$.
  Since $\delta \geq n$, by the pigeonhole principle, at least for one $k$, $S_k$ has more than $\delta/2\log n$ edges.
  Hence, for every cut $C \in \mathscr{C}$, there exists a $\delta$-good $k$-projection.  
  By Lemma~\ref{lemma:goodProjection}, for any $\delta \geq n$ and for any $k$ that is power of $2$ every $\delta$-good $k$-projection contains at least one sampled edge with probability at least $1-\frac{1}{n^3}$.
  Therefore, every $C \in \mathscr{C}$ contains at least one sampled edge w.h.p..
\end{proof}

\subsection{Sampling and Component Graph}
Let $\hat{E}$ be the set of sampled edges. 
We now prove the two properties of sampled edges $\hat{E}$ and the component graph $cg[G, \hat{E}]$ induced by $\hat{E}$. 

We say an edge $e = \{u, v\}$ is \textit{charged} to vertex $u$ if $rd(u) < rd(v)$.  
  If $rd(u) = rd(v)$ then the edge $\{u, v\}$ is charged to either $u$ or $v$ arbitrarily. 
  Hence every edge is charged to exactly one vertex.
\begin{lemma}\label{lemma:ssize}
  Let $G=(V,E)$ be an $n$-vertex graph  and let $\hat{E}$ be the set of edges obtained by independently sampling each edge $e\in E$ with probability $p_e = \min\{50\log^2 n/rd(e), 1\}$. 
  Then, with probability at least $1 -\frac{1}{n}$, 
  the number of sampled edges charged to a vertex is at most $150\log^2 n$ and 
  the total number of sampled edges is $O(n\log^2 n)$.
\end{lemma}
\begin{proof}
  Consider a vertex $u$.
  Let $E_u$ be the set of edges charged to $u$.
  Edges in $E_u$ are sampled with probability $\min\{50\log^2 n/rd(u), 1\}$. 
  If $rd(u) \leq 50\log^2 n$ then we are done, therefore assume $rd(u) > 50\log^2 n$. 
  Let $X_e$ be the indicator random variable indicating if edge $e \in E_u$ is sampled. 
  Let $X = \sum_{e\in E_u} X_e$ denote the total number of sampled edges charged to $u$. 
  Then $\E[X] = \sum_{e \in E_u} 50\log^2 n/rd(u) \leq \degree(u)\cdot50\log^2 n/rd(u) \leq 100\log^2 n$. 
  Since edges are sampled independently, by Chernoff's bound we have, 
  \[\Pr\left(X > 150\log^2 n\right) \leq \exp\left(-\frac{100\log^2 n}{12}\right) < \frac{1}{n^3}.\] 
  By applying a union bound over all vertices, for every vertex $v$ the probability that the number of sampled edges charged to $v$ exceeds $150\log^2 n$ is at most $\frac{1}{n^2}$. 
  Therefore, the total number of sampled edges is $O(n \log^2 n)$ w.h.p..
\end{proof}

\begin{lemma}\label{lemma:smallcut}
  Let $G=(V,E)$ be an $n$-vertex graph  and let $\hat{E}$ be the set of edges obtained by independently sampling each edge $e\in E$ with probability $p_e = \min\{50\log^2 n/rd(e), 1\}$. 
  Let $G'=(\mathcal{C}, E')$ be the component graph $cg[G, \hat{E}]$. 
  Then with probability at least $1 - \frac{1}{n}$,  $|E'| = O(n)$.
\end{lemma}
\begin{proof} 
  We first prove that the max-cut size of $G'$ is less than $n$. 
  Assume for the sake of contradiction that a cut $(M, \mathcal{C}\setminus M)$ in $G'$ where $M=\{C_{1},\ldots C_{m}\}  \subseteq \mathcal{C}$ has size more than $n$. 
  This means that the cut $(M'=\{v \mid v \in \cup_{i=1}^{m}C_{i}\}, V\setminus M')$ has size more than $n$ in $G$, since each edge in $E'$ is induced by of one or more edges from $E$.
  By Theorem~\ref{thm:cut}, w.h.p., $\hat{E}$ contains at least one edge from the cut $(M', V\setminus M')$ of $G$.
  Let $\{u, v\} \in \hat{E}$ such that $u \in M'$ and $v\in V\setminus M'$. 
  Hence $u$ has to be in one of $C_{i}, i \in [m]$ and let it be $C_{i'}$. 
  This implies in the component graph, $u$ and $v$ to be in the same component $C_{i'}$. Hence a contradiction. 
  This is true w.h.p. for any cut of size more than $n$ in $G'$. 
  Therefore the max-cut size of $G'$ is less than $n$ w.h.p.. 

  \noindent Consider the following randomized algorithm to find a cut in $G'$. 
  Each vertex in $G'$ is independently added to a set $U$ with probability $1/2$. 
  We now analyze the size of the cut $(U, \mathcal{C}\setminus U)$.  
  The probability that an edge crosses this cut is $1/2$. 
  Since there are $|E'|$ edges, the expected size of this cut is $|E'|/2$. 
  But we know that the max-cut size is at most $n$. Therefore, $|E'| < 2n$. 
\end{proof}

\noindent We summarize this section with the following theorem.  
\begin{theorem}\label{thm:sample}
  Let $G=(V,E)$ be an $n$-vertex graph  and let $\hat{E}$ be the set of edges obtained by independently sampling each edge $e\in E$ with probability $p_e = \min\{50\log^2 n/rd(e), 1\}$. 
  Let $cg[G, \hat{E}]$ be the component graph induced by $\hat{E}$. 
  Then with probability at least $1 - \frac{1}{n}$ we have, 
  \begin{enumerate}
    \item The number of sampled edges is $|\hat{E}| = O(n \log^2 n)$, 
    \item The number of inter-component edges, that is, the number of edges in $cg[G, \hat{E}]$ is $O(n)$.
  \end{enumerate}
\end{theorem}

\section{Connectivity Verification via Random Edge-Sampling} %in \texorpdfstring{$O(\log \log \log n)$}{O(log log log n)}  Rounds}
\label{sec:conver}
In this section, we describe how to utilize the degree-based edge sampling from the previous section to solve the Connectivity Verification problem on a Congested Clique.
This randomized algorithm solves the Connectivity Verification problem in $O(\log \log \log n)$ rounds w.h.p. by combining the degree-based edge sampling with the Lotker et al. deterministic MST algorithm~\cite{lotker2005mstJournal}. 
We first describe the  Lotker et al. MST algorithm~\cite{lotker2005mstJournal}.

The Lotker et al. algorithm runs in phases, taking constant number of communication rounds per phase.
At the end of phase $k \geq 0$, the algorithm has computed a partition $\mathcal{F}^k = \{{F_1}^k , {F_2}^k, \ldots,  {F_m}^k\}$ of the nodes of $G$ into clusters.
Furthermore, for each cluster $F \in \mathcal{F}^k$, the algorithm has computed a minimum spanning tree $T(F)$. 
It is worth noting that every node in the network knows the partition $\mathcal{F}^k$ and the collection $\{T (F) \mid F \in \mathcal{F}^k\}$ of trees.
It is shown that at the end of phase $k$ the size of the smallest cluster is at least $2^{2^{k-1}}$ and hence $|\mathcal{F}^k| \leq n/2^{2^{k-1}}$. 
In the following we refer to the Lotker et al. algorithm as the \textsc{CC-MST} algorithm. 
Let \textsc{CC-MST}$(G, k)$ denote the execution of \textsc{CC-MST} on graph $G$ for $k$ phases. 
\begin{theorem}[Lotker et al.\cite{lotker2005mstJournal}]\label{thm:lotker}
  \textsc{CC-MST} computes an MST of an $n$-node edge-weighted clique in $O(\log \log n)$ rounds. 
  At the end of phase $k$, \textsc{CC-MST} has computed a partition $\mathcal{F}^k = \{F_1^k , F_2^k, \ldots,  F_m^k\}$ and $\mathcal{T}^k= \{T(F) \mid F \in \mathcal{F}^k\}$ which has the following properties: 
  (i) $m \leq 2^{2^{k-1}}$, 
  (ii) Every node knows $\mathcal{F}^k$ and $\mathcal{T}^k$, and
  (iii) If the largest weight of an edge in a cluster $F_i^k$ is $w$, then there is no edge with weight $w' < w$ such that it has one end point in $F_i^k$ and other in $F_j^k$ for any $j\neq i$. 
\end{theorem}

Our connectivity verification algorithm runs in three phases. 
Initially, our graph can be viewed as having $n$ components - one for each vertex. 
In Phase~1 we reduce the number of components by running \textsc{CC-MST} for $O(\log \log \log n)$ phases. 
Phase~2 operates on the component graph induced by the edges selected in Phase~1  and samples edges from this component graph using degree-based probabilities as discussed in the earlier section. 
Phase~3 is executed on the component graph induced by edges selected in Phase 2. 
Each phase outputs a forest $\mathcal{T}$ and a component graph $G'$ induced by edges in $\mathcal{T}$, that is, 
at the end of each phase every node knows all the edges in $\mathcal{T}$ and knows which of the incident edges are the inter-component edges in $G'$.  
Given a subgraph of $G$ that is a tree, we call this tree \textit{finished} if it is a spanning tree of a connected component in the graph $G$.
A tree which is not finished is referred as \textit{unfinished} tree. 
Each phase construct trees, some of which might be unfinished. 
A finished tree need not play any further part in the algorithm.
We show that Phase~2 and Phase~3 run in $O(1)$ rounds each w.h.p. and 
at the end of Phase 3 all trees are finished. 

We make use of the following subroutine at the end of Phase 1 and Phase 2 to ``construct'' the component graph 
(refer to Subsection~\ref{sub:tech} for definitions and notations). 
The subroutine \textsc{BuildComponentGraph}$(G, \hat{E})$ takes a subset of edges of $\hat{E} \subseteq E$ as input and it is assumed that initially all nodes know all edges in $\hat{E}$ and components induced by $\hat{E}$.  
It returns the component graph $cg[G, \hat{E}]$. 
At the end of this subroutine each leader knows the inter-component edges incident on its component. 
This subroutine can be implemented in $O(1)$ rounds using \lra~as follows:
%Let $\mathcal{C}$ be the set of connected components of graph $(V, \hat{E})$.
%Let $c(u)$ denote the component of a node $u$. 
%Let $\ell(C)$ be the node with minimum ID in $C$ for each $C \in \mathcal{C}$. 
%We call $\ell(C)$ as the leader of component $C$. 
every node $v$ in a component $C$, for every incident edge $\{u, v\}$ such that $c(u) \neq C$ adds a message destined for $\ell(c(u))$ in the sending queue to notify $\ell(c(u))$ of the presence of the inter-component edge $\{C, c(u)\}$, if it already has such a message in the queue (due to a different incident edge $\{u', v\}, c(u')=c(u)$) then $v$ ignores this edge.
Hence the sending queue of each node contains at most $n$ messages (since there can be at most $n$ components). 
Each leader receives at most $n$ messages since each node is sending only a single message to a leader.
Therefore, we can use \lra~to route these messages in $O(1)$ rounds. 
After this step, every leader $\ell(C)$ knows the incident inter-component edges. 

Algorithm \textsc{ReduceComponents} describes Phase~1. 
\begin{algorithm}[H]
  \caption{Phase 1: \textsc{ReduceComponents} \label{algo:phase1}}
  \begin{algorithmic}[1]
    \REQUIRE A graph $G = (V, E)$. 
    \ENSURE  $\mathcal{T}_1$ - a spanning forest of $G$ with at most $n/\log^2 n$ unfinished trees and component graph induced by these edges  
    \STATE   Assign unit weights to edges in $G$ to obtain a weighted graph $G_w$; make $G_w$ a clique by adding edges not in $G$ and assign weight $\infty$ to these newly added edges.  
    \STATE   $(\mathcal{F}, \mathcal{T}_\infty) \leftarrow $ \textsc{CC-MST}$(G_w, \log \log \log n + 1)$
    \STATE $\mathcal{T}_1 \leftarrow \mathcal{T}_\infty \setminus \left\{\left\{u, v\right\} \in E(\mathcal{T}_\infty) \mid wt(u, v) = \infty \right\}$  
    \STATE $G_1 \leftarrow$ \textsc{BuildComponentGraph}$(G, \mathcal{T}_1)$
    \RETURN $(\mathcal{T}_1, G_1)$
  \end{algorithmic}
\end{algorithm}
\noindent Input to Algorithm \textsc{ReduceComponents} is a graph $G$. 
At the end of this algorithm, every node knows the ID of the leader of the component it belongs to and every leader knows incident inter-component edges in the component graph induced by edges selected during the execution. 
In Step~1, to every edge in the input graph $G$ we assign weight $1$; pairs of vertices not adjacent are assigned weight $\infty$. 
Step~2 simply executes \textsc{CC-MST} on this weighted clique for $\log \log \log n + 1$ phases which returns clusters $\mathcal{F}$ and a forest $\mathcal{T}_\infty$ of trees, one spanning tree per cluster. 
There might be few edges selected by \textsc{CC-MST} with weights $\infty$ and in Step~3 we discard these edges. 
By Theorem~\ref{thm:lotker}, every node knows $\mathcal{T}_\infty$ (so $\mathcal{T}_1$) and hence we can execute \textsc{BuildComponentGraph} (Step~4) in $O(1)$ rounds. 
At the end of \textsc{ReduceComponents} we have the following properties:
\begin{lemma}
  \label{lemma:phase1Finished}
 If a tree in $\mathcal{T}_\infty$ has an edge with weight $\infty$ then after removing this edge both of the obtained trees are either finished trees of $G$ or contains further $\infty$-weight edges. If we remove all the $\infty$-weight edges then all the newly obtained trees are finished trees of $G$. 
\end{lemma}
\begin{proof}
  Let $e = \{u, v\}$ be an edge in $T(F_{i^*})$  has weight $\infty$, i.e., $e\notin E$.  
  By Theorem~\ref{thm:lotker} (Property (iii)), all the incident edges on $F_{i^*}$ have weights $\infty$.
  In other words, there is no edge in $G$ which has exactly one endpoint in $F_{i^*}$. 
  If the trees obtained by removing $e$ does not contain any further $\infty$-weight edges then the both trees are finished trees. 
  If it contains $\infty$-weight edges then we repeat the above argument on the both of the trees until we obtain trees with no $\infty$-weight edges.
  By the earlier argument all these obtained trees are finished trees. 
\end{proof}
\begin{lemma}
  \label{lemma:phase1Unfinished}
  The number of unfinished trees in $\mathcal{T}_1$ are at most $\frac{n}{\log^2 n}$. 
\end{lemma}
\begin{proof}
  By Theorem~\ref{thm:lotker} (Property (i)), we have $|\mathcal{T}_\infty| = n/\log^2 n$. 
  By Lemma~\ref{lemma:phase1Finished}, by removing $\infty$-weight edges increases only the number of finished trees. 
  Therefore, the number of unfinished trees cannot be more than $n/\log^2 n$. 
\end{proof}
\noindent Also, it is easy to see that Phase~1 runs in $O(\log \log \log n)$ rounds, since Step~2 takes $O(\log \log \log n)$ rounds, but the rest take only $O(1)$ rounds each.   
\begin{comment}
\begin{lemma}
  \textsc{BuildComponentGraph}$(G, \mathcal{T})$ runs in $O(1)$ rounds in the Congested Clique.
\end{lemma}
\begin{proof}
  Please refer to the description of \textsc{BuildComponentGraph}. 
  By Lemma~\ref{lemma:f}, there are at most $n/\log^2 n$ different components and hence a node $v$ has at most $n/\log^2 n$ neighbors $u$ with distinct $c(u)$ values. 
  Hence the sending queue of each node has at most $n/log^2 n$ messages to send. 
  Every leader has to receive at most $n$ messages since each node sends at most one message to each leader. 
  Therefore by using Lenzen's routing algorithm we can route these messages in $O(1)$ rounds on the Congested Clique. 
\end{proof}
\end{comment}
 
Phase~2 runs on the component graph $G_1$ returned by Phase~1 and computes a spanning forest $\mathcal{T}_2$ of $G_1$ such that the component graph induced by $\mathcal{T}_2$ has at most $O(n)$ inter-component edges. 
Note that there might be some finished trees in $\mathcal{T}_1$. 
The components induced by finished trees do not have any incident edge in $G_1$ and hence the corresponding vertices will be isolated in $G_1$. 
Let $G'_1(V'_1, E'_1)$ be the graph obtained by removing isolated vertices from $G_1$. 
By Lemma~\ref{lemma:phase1Unfinished}, $|V'_1| \leq n/\log^2 n$. 
Let $v^*$ denote the vertex in $V$ with minimum ID.
\begin{algorithm}[H]
  \caption{Phase~2: \textsc{RemoveLargeCuts} \label{algo:phase2}}
  \begin{algorithmic}[1]
    \REQUIRE $G'_1(V'_1, E'_1)$ obtained by removing isolated vertices from $G_1$ where $V'_1 \subseteq V$ and $|V'_1| \leq \frac{n}{\log^2 n}$  
    \ENSURE  $\mathcal{T}_2$ - a spanning forest of $G'_1$ such that the number of edges in the component graph $G_2$ induced by $\mathcal{T}_2$ is $O(n)$.
    \STATE  $S \leftarrow \emptyset$. \\ For each edge $e = \{u, v\} \in E'_1$, add edge $e$ to $S$ with probability $\min\{1, \frac{50\log^2 n}{rd(e)}\}$ where $rd(e)$ is the rounded-degree of $e$ with respect to $G'_1$.
    \STATE  Gather $S$ at vertex $v^*$ (the node in $V$ with the lowest ID).
    \STATE  $v^*$ executes locally : $\mathcal{T}_2 \leftarrow $\textsc{SpanningForest}$(G'_1[S])$. 
    \STATE  $v^*$ assigns each edge in $\mathcal{T}_2$ to a node in $V$ such that each node is assigned a single edge and send edges to assigned nodes.  
    Each node in $V$ then broadcast the edge it received from $v^*$ so that all nodes now know $\mathcal{T}_2$.   
    \STATE  $G_2 \leftarrow $ \textsc{BuildComponentGraph}$(G'_1, \mathcal{T}_2)$ 
    \RETURN $(\mathcal{T}_2, G_2)$
  \end{algorithmic}
\end{algorithm}

Step~1 of Phase~2 can be implemented as follows. 
Each node in $G'_1$ broadcast its degree with respect to $G'_1$. 
An edge $\{u, v\}$ is ``charged'' to node $u$ if $rd(u) < rd(v)$ or $rd(u) = rd(v)$ and $ID(u) < ID(v)$. 
Node $u$ computes $rd(e)$ for each edge $e$ charged to it and then samples each $e$ independently  with probability $\min\{50\log^2 n/rd(e), 1\}$. 
Node $u$ constructs a queue of messages intended for node $v^*$ consisting of all edges it sampled. 
We show below that the contents of all these queues can be sent to $v^*$ in $O(1)$ rounds. 
Step~3 is a local step executed at $v^*$. 
Step~4 makes sure that each node knows all the components induced by sampled edges in Step~1. 
Therefore, Step~5 can be executed in $O(1)$ rounds. 
We now show that Step~2 can be implemented in $O(1)$ rounds by proving the following claim and then appealing to Lenzen's routing algorithm to route the messages in the queue of each node.  
\begin{lemma}
  \label{lemma:p2sample}
  $|S| = O(n)$ with probability at least $1 - \frac{1}{n}$.
\end{lemma}
\begin{proof}
  By Lemma~\ref{lemma:ssize}, the number of sampled edges charged to a node is $O(\log^2 n)$ w.h.p. and therefore, 
  if a graph has $n'$ vertices then the number of sampled edges is $O(n'\cdot \log^2 n)$ w.h.p.. 
   The graph $G'_1$ is the component graph induced by unfinished trees in Phase~1 and we showed that $n' = |V'_1| \leq  n/\log^2 n$. 
   Therefore, $|S| = O(n)$ w.h.p..
\end{proof}
\noindent Since each node's queue can contain at most $O(n/\log^2 n)$ messages w.h.p. and since $|S| = O(n)$ w.h.p., $v^*$ has to receive at most $O(n)$ messages w.h.p., and hence we can route these messages in $O(1)$ rounds w.h.p. by using Lenzen's routing algorithm. 

\begin{lemma}
  \label{lemma:phase2}
  Phase~2 runs in $O(1)$ rounds w.h.p. and returns a spanning forest $\mathcal{T}_2$ such that the component graph $G_2$ induced by $\mathcal{T}_2$ has $O(n)$ edges.
\end{lemma}
\begin{proof}
  The discussion just before Lemma~\ref{lemma:p2sample} proves that each step in Algorithm~\ref{algo:phase2} can be implemented in $O(1)$ rounds w.h.p.. 
  By Theorem~\ref{thm:sample}, the number of inter-component edges in $G_2$ is $O(n)$ w.h.p..
\end{proof}

\noindent We execute the Phase~3 on the component graph $G_2$ obtained in Phase~2.  
Let $G'_2(V'_2, E'_2)$ denote the graph obtained by removing isolated vertices from $G_2$. 
Phase~3 computes $\mathcal{T}_3$ -a spanning forest of $G'_2$. 
\begin{algorithm}[H]
  \caption{Phase 3: \textsc{HandleSmallCuts}}
  \begin{algorithmic}[1]
    \REQUIRE $G'_2(V'_2, E'_2)$ obtained by removing isolated vertices from $G_2$ where $V'_2 \subset V$ and $|E'_2| = O(n)$
    \ENSURE  $\mathcal{T}_3$ - a spanning forest of $G'_2$.
    \STATE   Gather $E'_2$ at vertex $v^*$ (the node in $V$ with the lowest ID).
    \STATE   $v^*$ locally executes: $\mathcal{T}_3 \leftarrow$ \textsc{SpanningForest}$(G'_2)$.  
    \STATE   $v^*$ assigns each edge to a node in $V$ such that each node is assigned a single edge and send edges to assigned nodes.  
	    Each node then broadcast the edge it received so that all nodes now know $\mathcal{T}_3$.
    \STATE   $G_3 \leftarrow$ \textsc{BuildComponentGraph}$(G'_2, \mathcal{T}_3)$
    \RETURN  $(\mathcal{T}_3, G_3)$. 
  \end{algorithmic}
\end{algorithm}
By Lemma~\ref{lemma:phase2}, the number of inter-component edges ($|E'_2|$) in the component graph $G'_2$ is $O(n)$ and the degree of each node in $G'_2$ is at most $O(n/\log^2 n)$. 
Therefore, Step~1 can be executed in $O(1)$ rounds using Lenzen's routing algorithm.
Algorithm \textsc{Conn} summarizes our algorithm. 
\begin{algorithm}[H]
  \caption{\textsc{Conn}\label{algo:conn}}
  \begin{algorithmic}[1]
    \REQUIRE $G(V,E)$
    \ENSURE  a maximal spanning forest of $G$
    \STATE   $(\mathcal{T}_1, G_1) \leftarrow $ \textsc{ReduceComponents}$(G)$
    \STATE   $(\mathcal{T}_2, G_2) \leftarrow $ \textsc{RemoveLargeCuts}$(G_1)$ 
    \STATE   $(\mathcal{T}_3, G_3) \leftarrow $ \textsc{HandleSmallCuts}$(G_2)$
    \RETURN  $\left\{\mathcal{T}_1 \cup \mathcal{T}_2 \cup \mathcal{T}_3\right\}$
  \end{algorithmic}
\end{algorithm}
\noindent We now prove that Algorithm \textsc{Conn} solves the Connectivity Verification problem. 
\begin{lemma}
   At the end of Algorithm \textsc{Conn} every node in $G$ knows a spanning forest with exactly as many trees as the number of connected components in $G$, that is, \textsc{Conn} returns a maximal spanning forest of $G$.
\end{lemma}
\begin{proof}
  Let $\mathcal{M} = \mathcal{T}_1 \cup \mathcal{T}_2 \cup \mathcal{T}_3$ denote the spanning forest returned by Algorithm \textsc{Conn}. 
  Let $C$ be the number of connected components in $G$. 
  Let $M$ be the number of maximal trees in $\mathcal{M}$ (which are the number of vertices in $G_3$). 
  We want to show that $C = M$. 

  Assume $C < M$.
  Therefore, there exists at least one edge $\{u, v\} \in E$ such that $u$ and $v$ are in the same connected component of $G$ but they are not in the same tree in $\mathcal{M}$. 
  Let $u$ is in tree $T_i \subset \mathcal{M}$ and $v$ is in tree $T_j \subset \mathcal{M}$, $i\neq j$, that is $T_i$ and $T_j$ are the unfinished trees.
  It means $u$ and $v$ were in the separate components at the end of Phase~1 and Phase~2. 
  But then in Phase~3, the edge $\{u, v\}$ is inspected and they will be in the same spanning tree computed by Phase~3. 
  Hence a contradiction.

  Now assume $C > M$. 
  This is possible only if we add an edge $e \notin E$ to $\mathcal{M}$ during the execution of \textsc{Conn}. 
  We only add additional edges in Phase~1 but we assign weight $\infty$ to edges which are not in $E$ and these edges are removed from $\mathcal{M}$.  
  Therefore the additional edge must have weight $\infty$ and hence not in $\mathcal{T}$ and won't be present in $\mathcal{M}$. 
  
  \noindent By combining the above two arguments we have $C = M$.
\end{proof}
\noindent This lemma establishes the correctness of our algorithm. 
The discussion in this section also shows that \textsc{Conn} runs in $O(\log \log \log n)$ rounds with high probability. 
We summarize the result of this section in the following theorem. 

\begin{theorem}
  Algorithm \textsc{Conn} solves the Connectivity Verification problem in $O(\log \log \log n)$ rounds with probability at least $1 - \frac{1}{n}$. 
\end{theorem}

\section{Exact MST via Random Edge-Sampling} %in \texorpdfstring{$O(\log \log \log n)$}{O(log log log n)}  Rounds}
\label{sec:mst}
In this section we show how to solve the MST problem on a Congested Clique using ideas from our Connectivity Verification algorithm.
Initially each node is a component and hence initially there are $n$ components. 
We first reduce the number of components to $n/\log^2 n$ using the Lotker et al. MST algorithm similar to Phase~1 of our Connectivity Verification algorithm in $O(\log \log \log n)$ rounds. 
Then we reduce the MST problem on this graph to two subproblems using a sampling lemma by Karger et al.~\cite{KKT1995MST}.
Each subproblem has to compute a MST of a weighted graph whose average degree is at most $\sqrt{n}$. 
We first show that this reduction can be completed in $O(1)$ rounds. 
Finally we show how to compute MST on a graph with average degree $\sqrt{n}$ and number of components $n/\log^2 n$ in $O(1)$ rounds.

\subsection{Reducing Components and Edges}
We first reduce the number of components to at most $n/\log^2 n$ components by executing \textsc{CC-MST} for $\log \log \log n + 1$ phases similar to our Connectivity Verification algorithm. 
Let $\mathcal{T}_1$ be the spanning forest and $G_1$ be the component graph obtained by executing the above step. 
By the property of \textsc{CC-MST}, $\mathcal{T}_1$ is a subset of a MST of $G$ (Theorem~\ref{thm:lotker}). 
Our goal now is to complete this MST by deciding which of the edges in $G_1$ are in the MST. 

Karger et al.~\cite{KKT1995MST} designed a randomized linear-time algorithm to find a MST in a edge-weighted graph in a sequential setting (RAM model).
A key component of their algorithm is a random edge sampling step to discard edges that cannot be in the MST. 
For completeness we state their sampling result and the necessary terminology.
\begin{definition}[$F$-light edge~\cite{KKT1995MST}] 
 Let $F$ be a forest in a graph $G$ and let $F(u,v)$ denote the path (if any) connecting $u$ and $v$ in $F$. 
 Let $wt_F(u, v)$ denote the maximum weight of an edge on $F(u, v)$ (if there is no path then $wt_F(u, v) = \infty$). 
 We call an edge $\{u, v\}$ is \emph{$F$-heavy} if $wt(u, v) > wt_F(u, v)$, and \emph{$F$-light} otherwise. 
\end{definition}
\noindent Karger et al.~\cite{KKT1995MST} proved the following lemma. 
\begin{lemma}[KKT Sampling Lemma~\cite{KKT1995MST}]
 \label{lemma:kkt}
Let $H$ be a subgraph obtained from $G$ by including each edge independently with probability $p$,
and let $F$ be the minimum spanning forest of $H$. 
The number of $F$-light edges in $G$ is at most $n/p$ w.h.p.. 
\end{lemma}
\noindent The implication of the above lemma is that if we set $p = 1/\sqrt{n}$ then the number of sampled edges in $H$ and the number of $F$-light edges in $G$ both are $O(n^{3/2})$ w.h.p.. 
Also, none of the $F$-heavy edges can be in a MST of $G$. 
Therefore if we compute a minimum spanning forest $F$ of $H$ then we can discard all the $F$-heavy edges and 
it is sufficient to compute a MST of graph induced by $F$-light edges in $G$. 
We have reduced the problem in two problems: (i) compute minimum spanning forest $F$ of $H$ where the number of edges in $H$ is $O(n^{3/2})$ w.h.p. and (ii) compute minimum spanning tree of the graph induced by $F$-light edges in $G$. 
Note that these two problems cannot be solved in parallel since the later problem depends on the output of the first problem. 

Algorithm~\ref{algo:kkt} summarizes our approach. 
In the beginning of Algorithm \textsc{Exact-MST} every node knows weights of incident edges and at the end of the execution every node knows all the edges that are in a MST computed by the algorithm. 
Algorithm \textsc{SQ-MST} computes a MST of a graph with $O(n/\log^2 n)$ vertices and $O(n^{3/2})$ edges and at the end of the execution of this algorithm, all nodes knows the MST computed by it.
In the next subsection we describe this algorithm and show that it runs in $O(1)$ rounds w.h.p.. 
\begin{algorithm}[H]
 \caption{\textsc{Exact-MST}\label{algo:kkt}}
 \begin{algorithmic}[1]
  \REQUIRE An edge-weighted clique $G(V, E)$
  \ENSURE  MST of $G$
  \STATE   $(\mathcal{T}_1, G_1) \leftarrow $\textsc{CC-MST}$(G, \log \log \log n + 1)$
  \STATE   $H \leftarrow $  a subgraph of $G_1$ obtained by sampling each edge independently with probability $\frac{1}{\sqrt{n}}$ 
  \STATE   $F \leftarrow $ \textsc{SQ-MST}$(H)$
  %\STATE   Gather $F$ at $v^*$ (the node with the lowest ID). \\
%	   $v^*$ sends each edge in $F$ to a node in $V$ such that no node in $V$ receives more than one edge.
%	   Every node then broadcast the message it received from $v^*$ so that all nodes now know $F$. 
  \STATE   $E_\ell \leftarrow \left\{\left\{u, v\right\} \in E(G_1) \mid \{u, v\} \mbox{ is } F\mbox{-light}\right\}$ 
  \STATE   $T_2 \leftarrow $ \textsc{SQ-MST}$(E_\ell)$
  \RETURN  $T_1 \cup T_2$
 \end{algorithmic}
\end{algorithm}

\subsection{Computing MST of $O(n^{3/2})$-size Graph}
In this subsection we show how to compute an MST of a subgraph $G'$ of $G$ with $O(n^{3/2})$ edges and $n/\log^2 n$ vertices  using our edge-sampling technique similar to the Connectivity Verification algorithm.

We have a graph $G'(V', E')$ ($V' \subset V, E' \subset E)$ with at most $n/\log^2 n$ vertices and $O(n^{3/2})$ edges  where $n$ is the number of nodes in the Congested Clique network $G$.
The bounds on number of vertices and number of edges are critical ensuring that our MST algorithm runs in $O(1)$ rounds. 
Algorithm \textsc{SQ-MST} (MST algorithm on a graph with average degree $\sqrt{n}$) describes our MST algorithm on subgraph $G'$. 
The high level idea is to sort the edges based on their weights, that is, 
each node needs to know the \emph{rank} $r(e)$ of each incident edge which is the index of $e$ in a global enumeration of the sorted edges.
This sorting problem can be solved in $O(1)$ rounds on the Congested Clique by using Lenzen's distributed sorting algorithm~\cite{lenzen2013routing}. 
Then each node partitions the incident edges based on their ranks. 
Thus we partition $E'$ into at most $\sqrt{n}$ sets $E_1, E_2, \ldots E_p$ ($p \leq \sqrt{n}$) each containing $n$ edges ($E_p$ might have less than $n$ edges) 
such that $E_1$ contains all the edges whose ranks are in the range $1$ to $n$, $E_2$ contains the edges with ranks between $n+1$ and $2n$, and so on. 
That is, each node knows the partition index of each incident edge. 

In the next step we gather set $E_i$ at a single \textit{guardian} node $g(i)$.
This can be done in $O(1)$ rounds as well because $|E_i| = O(n)$. 
The role of a guardian node $g(i)$ is to determine which of the edges in $E_i$ are a part of the MST. 
Specifically, $g(i)$ wants to know for each edge $e \in E_i$ whether there is a path between its endpoints in the graph induced by edges with ranks less than $r(e)$. 
That is, for each edge $e_j$ that a $g(i)$ has, $g(i)$ needs to find out whether there is a path between endpoints of $e_j$ in the graph induced by edges  $\left\{\cup_{k=0}^{i-1} E_k\right\} \cup \{ e_\ell \in E_i \mid r(e_\ell) < r(e_j)\}$. 
Thus each $g(i)$ needs to know a spanning forest of the graph $G_i$ induced by edges $\cup_{j < i} E_j$, and we show that it can be computed by executing the similar steps as Phase 2 and 3 of the Connectivity Verification algorithm on $G_i$ in $O(1)$ rounds.
Steps~\ref{algo:sqmst:2s} to~\ref{algo:sqmst:2e} of Algorithm \textsc{SQ-MST} are similar to \textsc{RemoveLargeCuts} procedure of Algorithm \textsc{CONN}(Algorithm~\ref{algo:conn}) which samples the incident edges (with respect to $G_i$) with probability based on its degree (with respect to $G_i$). 
%Note that here an edge $\{u, v\}$ is considered sampled if node $u$ samples it or node $v$ samples it and hence the probability of sampling an edge is slightly greater than the probability of sampling an edge as stated in our Sampling Theorem (Theorem~\ref{thm:sample}) but the result of the Sampling Theorem still holds.
Step~\ref{algo:sqmst:3} is similar to \textsc{HandleSmallCuts} which simply gathers the inter-component edges at the guardian node and process it locally. 
There are $\sqrt{n}$ such guardians - one for each partition $E_i$ and hence the challenge is executing $\sqrt{n}$ instances of these steps in parallel on the Congested Clique network. 
What helps in showing that these $p \leq \sqrt{n}$ instances can be executed in parallel is that $G'$ has at most $n/\log^2 n$ vertices and $O(n^{3/2})$ edges. 
The procedure \textsc{RouteLabels} implements this parallel execution in $O(1)$ rounds w.h.p..
We describe this procedure along with how to run all these steps in parallel in the next subsection. 
\begin{algorithm}[t]
  \caption{\textsc{SQ-MST} \label{algo:sqmst}}
  \begin{algorithmic}[1]
    \REQUIRE a weighted subgraph $G'(V', E', wt)$ with $\frac{n}{\log^2 n}$ vertices and $O(n^{3/2})$ edges
    \ENSURE  an MST of $G'$
    \STATE   $r(E) \leftarrow$ \textsc{DistributedSort}$(E)$ in non-decreasing order of edge-weights. 
    \STATE   Partition edges in $E$ based on their ranks $r(e)$ into $p$ partitions $E_1,E_2,\ldots E_p$ ($p\leq \sqrt{n}$),
	     each partition having $n$ edges ($E_p$ might have less than $n$ edges) such that 
	     $E_1$ contains edges with ranks $1, 2, \ldots, n$; $E_2$ contains edges with ranks $n+1, n+2, \ldots, 2n$; and so on.
    \STATE   Let $g(i)$ be the node in $G$ with ID $i$ and
	     assign $g(i)$ as the guardian of partition $i$. \\
	     Gather partition $E_i$ at $g(i)$.
    \FOR{$i=1$ \TO $i=p$ \textbf{in parallel}}
        \STATE Let $G_i = (V', \cup_{j=1}^{i-1} E_j)$.\label{algo:sqmst:2s} \\ 
	       $S_i \leftarrow \emptyset$.
	\STATE Each node $v$ executes this: for each incident edge $e \in \cup_{j=1}^{i-1} E_j$, 
	       adds edge $e$ to $S_i$ with probability $\min\left\{1, \frac{50\log^2 n}{rd(v)}\right\}$ where
	       $rd(v)$ is the rounded-degree of node $v$ with respect to $G_i$. 
	\STATE Gather $S_i$ at $g(i)$.
	\STATE $g(i)$ executes locally: $\mathcal{T}_i \leftarrow \textsc{SpanningForest}(G_i[S_i])$
	\STATE $g(i)$ informs each $v\in G_i$ about its component label $c_i(v)$ induced by $\mathcal{T}_i$. \label{algo:sqmst:2m}  
	\STATE Execute \textsc{RouteLablesAndInterComponentEdges}. It does the following: \\
	       (a). Identifies the inter-component edges ($\hat{E}_i$)  in $cg[G_i, \mathcal{T}_i]$. \\
	       (b). Gather $\hat{E}_i$ at $g(i)$.  \label{algo:sqmst:2e}
	\STATE $g(i)$ executes locally: \\ \label{algo:sqmst:3}
	       (a). $\mathcal{T'}_i \leftarrow \textsc{SpanningForest}(G_i[\hat{E}_i])$ \\ 
	       (b). $g(i)$ processes edges in $E_i$ in rank-based order. \\
		    For each edge $e_j=\{u, v\}$ in $e_1, e_2, \ldots$ : \\
		    if there is path between $u$ and $v$ in $\mathcal{T}_i \cup \mathcal{T'}_i \cup \{e_\ell \mid \ell < j\}$ then \\
		    discard $e_j$ else add $e_j$ to $\mathcal{M}_i$.  
    \ENDFOR
    \RETURN  $\cup_{i=1}^{p}\mathcal{M}_i$
  \end{algorithmic} 
\end{algorithm}

\subsection{Parallel Execution in Algorithm \textsc{SQ-MST}}
In this subsection we show that the for-loop on Line~\ref{algo:sqmst:2s}-\ref{algo:sqmst:3} of Algorithm~\ref{algo:sqmst} can be implemented in $O(1)$ rounds w.h.p..

Consider Lines~\ref{algo:sqmst:2s}-\ref{algo:sqmst:2m}. 
In these steps we sample incident edges as described in the algorithm.
We can gather $S_i$ at $g(i)$ for each $i$ in $O(1)$ rounds in parallel as follows: 
each node has $O(p\cdot\log^2 n)$ sampled edges (Lemma~\ref{lemma:ssize}) over all $p$ execution. 
There are at most $n/\log^2 n$ vertices in $G'$ therefore, $|S_i| = O(n)$ for each $i$. 
Hence each node needs to send $O(p\cdot \log^2 n) = O(n)$ messages and each guardian is a receiver of $O(n)$ messages.
Therefore we can deliver these messages in $O(1)$ by appealing to \lra. 

In Line~7, $g(i)$ locally computes a spanning forest $\mathcal{T}_i$ induced by edges in $S_i$. 
Let $c_i(v)$ denote the label of the component in $\mathcal{T}_i$, $v$ belong to. 
For each $v$, $g(i)$ sends $c_i(v)$ to $v$ (Line~8) and this can be done for all $i$ in parallel. 
Now each $v$ posses a $p$-size vector $\vec{C}(v)=\left(c_1(v), c_2(v), \ldots, c_p(v)\right)$ consisting of labels obtained from each guardian.
Let $\hat{E}_i$ denote the inter-component edges in the component graph $cg[G_i, \mathcal{T}_i]$.
The goal of \textsc{RouteLabels} is to identify which of the incident edges on $v$ are in $\hat{E}_i$ and gather $\hat{E}_i$ at $g(i)$ for each $i$ in parallel in $O(1)$ rounds w.h.p..   
Recall that an edge $\{u, v\} \in G_i$ is an inter-component edge in $cg[G_i, \mathcal{T}_i]$ if and only if $c_i(u) \neq c_i(v)$. 
This goal is similar to the goal of procedure \textsc{BuildComponentGraph} but here we need to do it for $p$ different instances in parallel. 
Notice that this is a non-trivial task since each node has a $p$-size label vector $\vec{C}(v)$ and there can be as many as $\Omega(\sqrt{n})$ neighbors to which this vector has to be delivered in order to identify edges in $\hat{E}_i$. 
We describe how to do a careful load-balancing to identify edges in $\hat{E}_i$ for all $i$ in parallel with the help of \emph{supporter} nodes. 

Partition $V$ into $\{sup(v) \mid v \in V'\}$ where $|sup(v)| = \left \lfloor \frac{\degree(v)}{\rho\sqrt{n}} \right \rfloor$ and $\degree(v)$ is the degree of $v$ with respect to $G'$ and $\rho>1$ is constant such that $|E'| \leq \rho\cdot n^{3/2}$. 
We call nodes in $sup(v)$ as \emph{supporter} nodes of $v$.
Such a partition exists because $\sum_{v\in V'}\degree(v) = 2|E'| \leq \rho n^{3/2}$ for a suitable constant $\rho$. 
Each $v \in V'$ informs all nodes in $sup(v)$ about its $p$-size label vector $\vec{C}(v)$. 
This can be done in $O(1)$ rounds by using \lra:
each node $v$ has $p$ labels to send to at most $\sqrt{n}$ supporter nodes, that is, $O(n)$ messages to send and each supporter node is a receiver of $p$ messages. 
The next task is to distribute the incident edges on $v$ to nodes in $sup(v) = \left\{s^v_1, s^v_2, \ldots s^v_{|sup(v)|}\right\}$. 
Let $E'(v)$ denote edges incident on $v$ in the graph $G'$. 
Partition $E'(v)$ into size-($\rho\sqrt{n} + 1)$ parts. 
Hence there are at most $|sup(v)|$ such parts and let these parts are $E'_1(v), E'_2(v), \ldots, E'_{|sup(v)|}(v)$. 
We can send part $E'_k(v)$  to supporter node $s^v_k$ in $O(1)$ rounds for all $k = 1, 2, \ldots, |sup(v)|$. 
Let $sup_u(v)$ denote the ID of the supporter in $sup(v)$ to which edge $\{u, v\}$ is sent. 
Each node $v$ sends a message to its neighbor $u$ informing about $sup_u(v)$ so that node $u$ knows that $v$ assigned edge $\{u, v\}$ to $sup_u(v)$. 
Then each node $v$ notifies $sup_u(v)$ about $sup_v(u)$ for each incident edge $\{u, v\}$. 
At this stage, each supporter $s^v_k$ knows the supporter of the end-point of all edges in $E'_k(v)$.  
Now to decide which of the edges in $E'_k(v)$ are in $\hat{E}_i$, node $s^v_k$ needs to know $\vec{C}(u)$ for all $u$ such that $\{u , v\} \in E'_k(v)$. 
Node $s^v_k$ requests this information to the corresponding supporter node, that is, for each edge $\{u, v\} \in E'_k(v)$, node $s^v_k$ requests $sup_v(u)$ to send $\vec{C}(u)$. 
Since $|E'_k(v)| = O(\sqrt{n})$, $s^v_k$ needs to receive $O(\sqrt{n}\cdot p)$ messages. 
On the other hand, it needs to send $O(\sqrt{n}\cdot p)$ messages in total. 
Hence this communication can be done in $O(1)$ rounds since $p \leq \sqrt{n}$ using \lra.
At this stage each $s^v_k$ has the necessary information to decide which of the edges in $E'_k(v)$ are in $\hat{E}_i$. 
For each edge in $E'_k(v)$ if it belongs to $\hat{E}_i$ then $s^v_k$ sends this edge to $g(i)$. 
There are $O(\sqrt{n})$ edges in $E'_k(v)$ and there are at most $p \leq \sqrt{n}$ different values of $i$, hence $s^v_k$ has $O(n)$ messages to send. 
Each guardian $g(i)$ needs to receive $\hat{E}_i$ which is of size $O(n)$ by our Sampling Theorem. 
Hence this communication can be done in $O(1)$ rounds. 
We summarize the above description in the Algorithm \textsc{RouteLabels} below.
\begin{algorithm}[H]
  \caption{\textsc{RouteLabels}}
  \begin{algorithmic}[1]
    \REQUIRE Each $v \in V'$ knows the $p$-size label vector $\vec{C}(v)$
    \ENSURE  Each guardian $g(i)$ for $i=1, 2, \ldots, p$ should know inter-component edges $\hat{E}_i$ 
    \STATE   Each vertex $v \in G'$ broadcast its degree $\degree(v)$. Let $\sum_{v \in V'} \degree(v) = \rho n^{\frac{3}{2}}$. 
    \STATE   Each vertex $v$ deterministically (all vertices use the same scheme) partitions $V$ ($n$ nodes)  into $|V'| = \frac{n}{\log^2 n}$ partitions: $\left\{sup(v) \mid v \in V'\right\}$ where
	    $|sup(v)| = \left\lfloor\frac{\degree(v)}{\rho\sqrt{n}}\right\rfloor$. \\
	     $sup(v) = \left\{s^v_1, s^v_2, \ldots, s^v_{|sup(v)|}\right\}$ is the set of \emph{supporter} nodes of $v$. 
    \STATE   Each $v$ sends the $p$-size vector $\vec{C}(v)$ to all nodes in $sup(v)$. 
    \STATE   Let $E'(v)$ be the set of incident edges on $v$ in $G'$. \\
	     Each $v$ partitions $E'(v)$ into $|sup(v)|$ partitions $E'_1(v), E'_2(v), \ldots E'_{|sup(v)|}(v)$ and
	     sends partition $E'_k(v)$ to $s^v_k$ for $k=1, 2, \ldots, |sup(v)|$. 
    \STATE   Let $sup_u(v) \in sup(v)$ denote the node to which $v$ sent edge $\{u, v\}$.\\
	     For each incident edge $\{u, v\}$, $v$ sends a message to $u$ notifying about $sup_u(v)$. 
    \STATE   For each $v \in V'$ and $k=1,2,\ldots, |sup(v)|$, each $s^v_k$ executes the following steps in parallel: 
    \bindent
        \STATE $s^v_k$ sends the $p$-size vector $\vec{C}(v)$ to all $u$ such that $\{u, v\} \in E'_k(v)$.
	\STATE For each $\{u, v\} \in E'_k(v)$ and for each $i=1, 2, \ldots, p$:\\
	\IF {$c_i(v) \neq c_i(u)$} \STATE $s^v_k$ sends $\{u, v\}$ to $g(i)$ \ENDIF 
    \eindent
  \end{algorithmic}
\end{algorithm}
The above discussion shows that each step of Algorithm \textsc{RouteLabels} can be implemented in $O(1)$ rounds w.h.p..

\begin{theorem}
  Algorithm \textsc{Exact-MST} computes a MST of a weighted clique in $O(\log \log \log n)$ rounds with probability at least $1- \frac{1}{n}$ on the Congested Clique. 
\end{theorem}

\bibliography{mymst}
\end{document}